\documentclass[sigconf]{acmart}
\usepackage{tabularx}
\usepackage{algorithm}
\usepackage{multirow}
\usepackage{enumitem}
\usepackage{bm}
\usepackage{xspace}
\usepackage{marvosym}
\usepackage[normalem]{ulem}
\usepackage{makecell} 
\usepackage{graphicx} 
\usepackage{subcaption} 
\usepackage{pifont}
\usepackage{tikz} 
\usepackage{booktabs}

\usepackage{bbding}
\usepackage{algorithm}  
\usepackage{algpseudocode}
\newtheorem{definition}{Definition}
\newtheorem{proposition}{Proposition} 
\newtheorem{corollary}{Corollary}

\newtheorem{counterexample}{Counterexample}
\AtBeginDocument{%
  }

\setcopyright{acmlicensed}
\copyrightyear{2018}
\acmYear{2018}
\acmDOI{XXXXXXX.XXXXXXX}
\acmConference[Conference acronym 'XX]{Make sure to enter the correct
  conference title from your rights confirmation email}{June 03--05,
  2018}{Woodstock, NY}
\acmISBN{978-1-4503-XXXX-X/2018/06}

\begin{document}

\title{GemiRec: Interest Quantization and Generation for Multi-Interest Recommendation}
\author{Zhibo Wu, Yunfan Wu, Quan Liu, Lin Jiang, Ping Yang, Yao Hu}
\email{{wuzhibo, wuyunfan, liuquan, jianglin, jiadi, xiahou}@xiaohongshu.com}
\affiliation{%
  \institution{Xiaohongshu  Co., Ltd}
  \country{Beijing, China}
}

\newcommand{\codebook}{\textit{Interest Dictionary}}
\newcommand{\distribution}{Multi-Interest Posterior Distribution Module}
\newcommand{\prediction}{Multi-Interest Retrieval Module}
\newcommand{\codebookm}{Interest Dictionary Maintenance Module}
\newcommand{\ds}{MIPDM}
\newcommand{\ps}{MIRM}
\newcommand{\cs}{IDMM}
\newcommand{\rec}{GemiRec}
\newcommand{\scb}{sub-dictionary}
\newcommand{\scbs}{sub-dictionaries}

\newcommand{\tightparagraph}[1]{%
  \vspace{0.4em}%
  \textit{#1}\hspace{0.6em}%
}


\received{20 February 2007}
\received[revised]{12 March 2009}
\received[accepted]{5 June 2009}

\setlength{\textfloatsep}{1.6pt plus 1pt minus 1pt}
\setlength{\intextsep}{1.6pt plus 1pt minus 1pt}
\setlength{\floatsep}{1.6pt plus 1pt minus 1pt}
\setlength{\abovecaptionskip}{1.6pt plus 1pt minus 1pt}
\setlength{\belowcaptionskip}{1.6pt plus 1pt minus 1pt}
\setlength{\abovedisplayskip}{1.6pt plus 1pt minus 1pt} 
\setlength{\belowdisplayskip}{1.6pt plus 1pt minus 1pt} 


\begin{abstract}
Multi-interest recommendation has gained attention, especially in industrial retrieval stage. 
Unlike classical dual-tower methods, it generates multiple user representations instead of a single one to model comprehensive user interests. 
However, prior studies have identified two underlying limitations: 
The first is interest collapse, where multiple representations homogenize. The second is insufficient modeling of interest evolution, as they struggle to capture latent interests absent from a user's historical behavior.
We begin with a thorough review of existing works in tackling these limitations. Then, we attempt to tackle these limitations from a new perspective. Specifically, we propose a framework-level refinement for multi-interest recommendation, named \rec.
The proposed framework leverages \textit{interest quantization} to enforce a structural interest separation and \textit{interest generation} to learn the evolving dynamics of user interests explicitly.  
It comprises three modules:
(a) \codebookm{} (\cs) maintains a shared quantized interest dictionary.
(b) \distribution{} (\ds{}) employs a generative model to capture the distribution of user future interests.
(c) \prediction{} (\ps) retrieves items using multiple user-interest representations. 
Both theoretical and empirical analyses, as well as extensive experiments, demonstrate its advantages and effectiveness.
Moreover, it has been deployed in production since March 2025, showing its practical value in industrial applications.
\end{abstract}




\maketitle

\section{Introduction}
\label{sec: introduction}
\begin{figure}[h]
    \centering
    \includegraphics[width=0.9\linewidth]{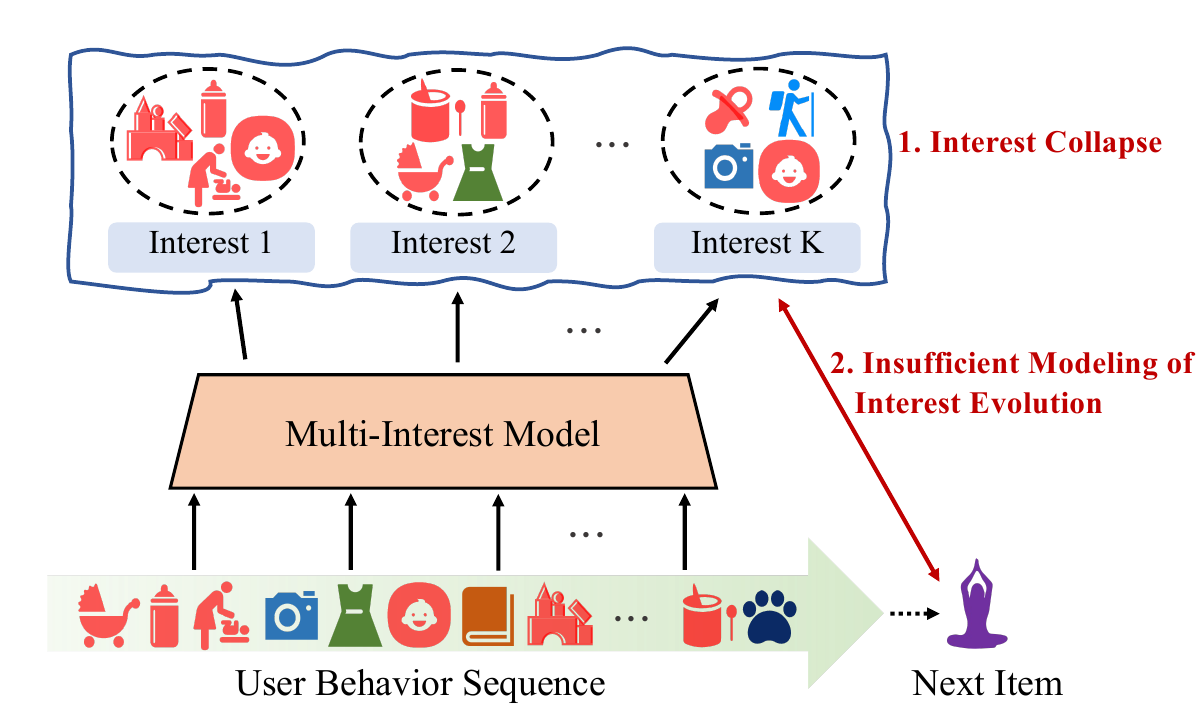}
    \caption{\small Illustration of the underlying challenges: (1). Interest collapse, where multiple interest representations primarily retrieve maternal and baby items (pink). (2). Insufficient modeling of interest evolution, where the learned interests fail to retrieve yoga-related items (purple) that are absent from the user's behavior sequence.}
    \label{fig:interpretation}
\end{figure}
Recommender systems have become an essential component of online platforms\cite{hwangbo2018recommendation, xie2022decoupled, zhang2023efficiently, zhou2023equivariant}. 
The industrial recommendation system is typically divided into four stages: retrieval, pre-ranking, ranking, and re-ranking. 
In the retrieval stage, which deals with billions of candidates, dual-tower architectures are widely adopted for their efficiency~\cite{covington2016deep, yi2019sampling}.
However, the single user representation may not fully capture the diverse nature of a user's interests~\cite{ohsaka2023curse,weller2025theoretical}, i.e, it tends to capture the dominant interest while neglecting others.

 
To address this limitation, multi-interest recommendation has been introduced. It generates multiple user representations, each representing distinct aspects of a user's interests, and updates the one with the maximum dot product to the target item.
Notable models MIND \cite{li2019multi} and ComiRec \cite{cen2020controllable} advanced this direction via dynamic routing and self-attention to extract multiple user interests. 

Despite the advancements of multi-interest recommendations, there still exist some underlying limitations. The first is interest collapse~\cite{li2022improving,du2024disentangled, zhang2022re4,xie2023rethinking,tan2021sparse}, where multiple user-interest representations homogenize, losing their distinctiveness. The second is insufficient modeling of interest evolution~\cite{wang2022target,tian2022multi, chen2021exploring}, which is indicated by the empirically observed performance degradation in capturing latent interests absent from user’s historical behavior.
We provide an intuitive illustration of the above limitations in Figure~\ref{fig:interpretation}.

To mitigate the interest collapse, prior studies, including CMI, Re4, SINE, REMI, SimRec, and DisMIR~\cite{li2022improving, zhang2022re4, tan2021sparse, xie2023rethinking, liu2024attribute, du2024disentangled}, made progress through
additional regularization to differentiate user-interest representations or learnable prototypes. However, they serve as a soft constraint rather than a structural separation; therefore, it does not preclude overlap. Meanwhile, overly strong soft constraints lead to the degradation of prediction accuracy~\cite {liu2022improving,goodfellow2016deep}. 
Building on these observations, we attempt to tackle this from a different perspective. 
We leverage quantization to assign each item to a certain category in an \codebook{}, serving as strict non-overlapping item clustering. In this way, semantically representative item embeddings amplify the structured semantic separation~\cite{lloyd1982least, arthur2006k}. This separation is maintained during training, where the positive item is associated with its corresponding category. Thus, what we need to know during inference is which k interest categories in the \codebook{} the user is interested in at the next time step.
By construction, we provide a theoretical analysis in Section \ref{sec: Theoretical Analysis}
that the quantization indeed induces a Voronoi partition~\cite{de2008computational}. Building on the property, it provides a non-trivial lower-bound of interest separation in principle, which can propagate to the user-interest representations under mild conditions, while the soft constraint regularization offers no such lower-bound.


To enhance the interest evolution modeling, prior studies such as PIMI~\cite{chen2021exploring}, MGNM~\cite{tian2022multi}, and TiMiRec~\cite{wang2022target}, made progress by incorporating temporal dynamics or soft-label distillation to strengthen contextualized interest modeling.
However, the bottleneck remains for two reasons:
First, as interest generation is integrated within the user tower, the latency restricts specific model design for capturing interest evolution. 
Second, interest generation is solely trained by the recommendation task, lacking an explicit objective to guide future interest learning. 
Motivated by these issues, we introduce a generative model decoupled from the user tower, which explicitly learns evolving interests through an independent next-interest prediction task.
This decoupled design, combined with a user-interest cache, offers flexibility in model complexity for interest generation, while essentially not increasing the inference latency.


Overall, to address the two underlying limitations of existing multi-interest recommendation methods, we propose a \underline{Ge}nerative \underline{M}ulti-\underline{I}nterest \underline{Rec}ommendation framework (named \rec{}). 
The framework centers on the \textbf{quantization and generation} of user interests. It consists of three modules: 
(a) \codebookm{} (\cs{}) maintains a shared vector-quantized \codebook{} containing discrete interest embeddings. 
(b) \distribution{} (\ds{}) employs a generative model
to capture the distribution of user interests at the next time step. 
(c) \prediction{} (\ps{}) retrieves items using multiple user-interest representations.

The main contributions of this work are:
\begin{itemize}[leftmargin=*, topsep=2pt, partopsep=0pt, itemsep=0pt, parsep=0pt]
    \item We propose a different perspective on tackling interest collapse and evolution through interest quantization and generation. Both theoretical and empirical analyses demonstrate its advantage. Further designed metrics validate the effectiveness.
    \item We provide practical guidance on further proposed joint optimization for interest quantization to better adapt its semantic separation to downstream recommendation tasks.
    \item For online deployment, we propose a top-$K$ user-interest indices cache to eliminate additional inference latency. 
    \item We conduct comprehensive experiments and online A/B tests, demonstrating its effectiveness. Furthermore, it has been successfully deployed in production, showing its practical value.
\end{itemize}

\section{Related Works}
\subsection{Vector Quantization}
Vector Quantization (VQ)~\cite{buzo1960speech} compresses a continuous representation space into a compact codebook, where each vector is approximated by a discrete code. 
Over time, advanced methods, such as product quantization~\cite{ge2013optimized, sabin2003product} and residual quantization~\cite{gray1998quantization, juang1982multiple, martinez2014stacked}, have been developed to reduce decoding errors.
A fundamental challenge of VQ is its inherently non-differentiable nature.
To address this, following VQ-VAE~\cite{van2017neural}, numerous studies~\cite{kang2020learning, lee2022autoregressive} have adopted the Straight-Through Estimator (STE)~\cite{bengio2013estimating} to enable gradient-based learning.
Recently, VQ has seen growing employment in recommendation~\cite{hou2023learning, rajput2024recommender}. 
However, the exploration of its potential for multi-interest modeling is limited.
\subsection{Multi-interest Recommendation}
Multi-interest was initially proposed in works~\cite{li2005hybrid, wang2007multi}, and gained traction after \text{MIND}~\cite{li2019multi} and \text{ComiRec}~\cite{cen2020controllable}, resulting in successive emergence of variant works.
With reference to recent survey~\cite{li2025multi}, we categorize multi-interest recommendations by concerns into interest collapse~\cite{li2022improving, xie2023rethinking, du2024disentangled,zhang2022re4,tan2021sparse,liu2024attribute}, cold start and fairness~\cite{tao2022sminet, zhang2022diverse}, interest evolution modeling~\cite{chen2021exploring, tian2022multi, wang2022target}, and others~\cite{wang2021popularity,chai2022user,liu2023co, shen2024multi}. 

Several works such as CMI, Re4, SINE, REMI, SimRec, and DisMIR~\cite{li2022improving, zhang2022re4, tan2021sparse, xie2023rethinking, liu2024attribute, du2024disentangled} made progress in interest collapse through additional regularization on interest embeddings or learnable prototypes, but they impose essentially a soft constraint rather than a structural separation, thus it does not preclude overlap between
retrieval sets of different interests. Meanwhile, overly strong soft constraints lead to the degradation of prediction accuracy~\cite {goodfellow2016deep}. 
Besides, works like PIMI, MGNM, and TiMiRe~\cite{chen2021exploring, tian2022multi, wang2022target} enhanced the contextualized interest distribution modeling by incorporating temporal dynamics or soft-label distillation.
While effective, it lacks specified model design and explicit supervision for modeling interest evolution, which may limit the capacity to sufficiently capture the dynamic user preferences over time.

\section{Methods} \label{sec: methods}
In this section, we introduce the methodology in detail. 
Table \ref{tb:symbol} summarizes the Mathematical symbols used throughout the paper.
\begin{figure*}[htbp]
    \centering
    \includegraphics[width=0.95\textwidth]{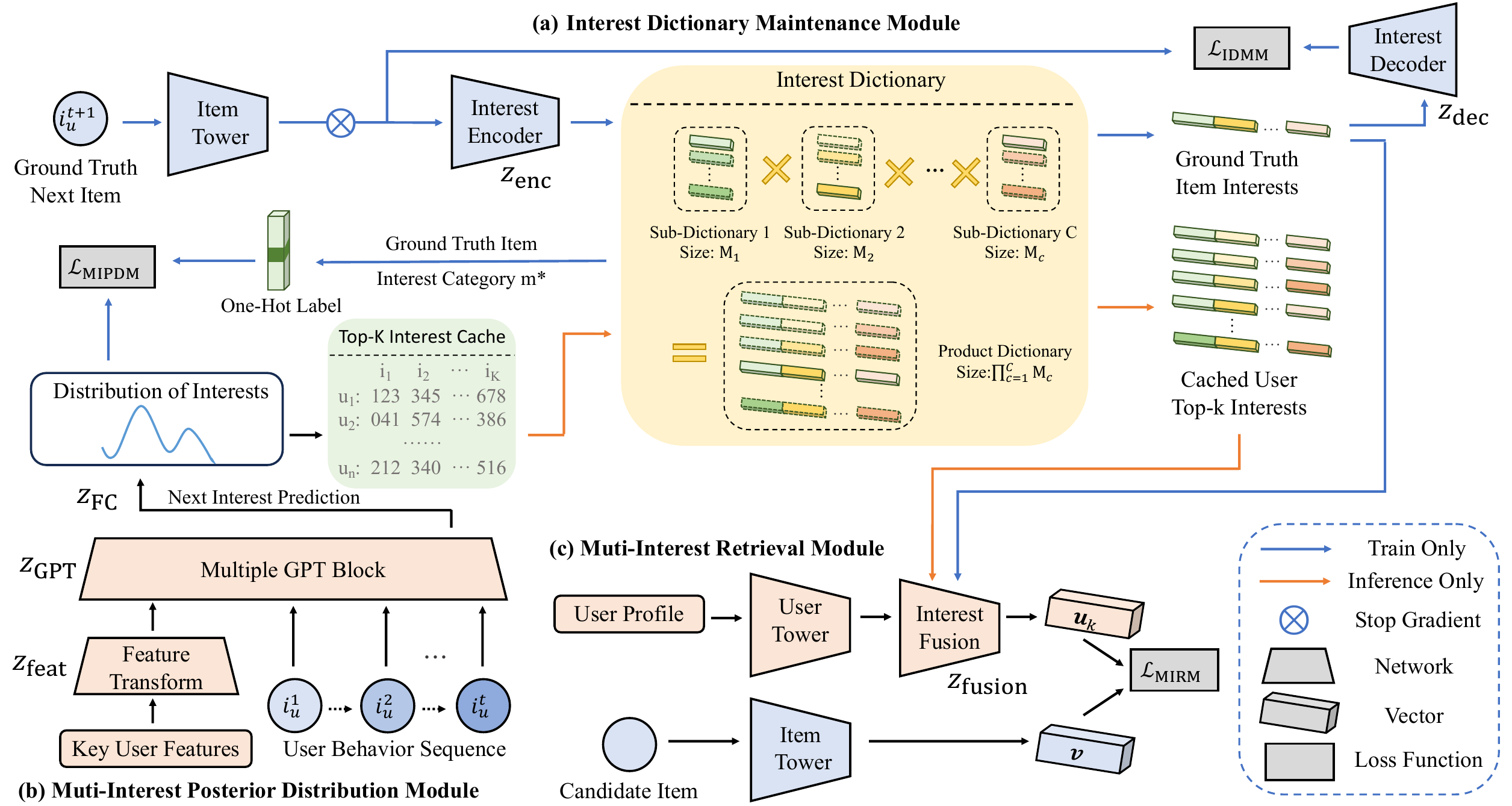}
    \caption{\small Overview of the \rec, illustrating the integration and interaction between the \cs, \ds, and \ps ~. 
    (a) \cs: maintaining multiple vector-quantized \scbs{} containing discrete interest embeddings; (b) \ds: employing a decoupled generative model to capture the distribution of user interests at the next time step; (c) \ps: retrieving items using multiple user-interest representations.}
    \Description{This figure illustrates the integration and interaction between the \cs, \ds, and \ps ~modules. \cs: maintaining a vector-quantized dictionary containing multiple sets of discrete interest embeddings; \ds: employing a generative model to capture the distribution of user interests at the next time step; \ps: retrieving items using multiple user-interest representations.
    }
    \label{fig: overall}
\end{figure*}


\subsection{Overview}
\subsubsection{Task Formulation}
We present the task of multi-interest recommendation.
Let \(\mathcal{U}\) and \(\mathcal{I}\) denote the user and item set. 
 Each user \(u \in \mathcal{U}\) has a interaction sequence \(\mathcal{I}_u = \{i_u^1, i_u^2, \ldots, i_u^t\}\), where \(i_u^t \in \mathcal{I}\) is the item interacted with at time \(t\).
The task aims to retrieve items based on k user-interest representations, denoted as $u_k$.
Specifically, the model calculates score \(\hat{y_k}(u, i)\) for each $u_k$:
\begin{equation}
\label{eq:multi-interest-score}
\hat{y_k}(u, i) = (\bm{u}_k)^\top \bm{v}_i,
\end{equation}
where $\bm{v}_i$ is the representation of item \(i\). In real-world applications, the above process is performed using nearest neighbor algorithms (e.g., Faiss~\cite{johnson2019billion}) to efficiently retrieve items 
 with respect to each $\bm{u}_k$. 

\subsubsection{Overall Architecture}
The overall architecture of \rec~is presented in Figure~\ref{fig: overall}, illustrating the integration and interaction between \cs,\ds, and \ps, including both training and inference.
In our framework, \cs~and \ds~ handle the interest quantization and generation, respectively, while \ps~is responsible for multi-interest recommendation. 


\subsection{\codebookm}
In this section, we provide details on how we construct an \codebook{} through interest quantization.

We implement RQ-VAE~\cite{zeghidour2021soundstream} to maintain an \codebook. 
The entire \codebook{} denoted as $\bm{E}^{*}$ can be seen as a combination of \(C\) \scbs. 
The c-th \scb~ is denoted as \(\bm{E}^{c} \in \mathbb{R}^{M_c \times d}\), where \(M_c\) represents \scb{} size, and \(d\) the dimension of each interest embedding. 
The total $C$ \scbs~form a Cartesian product to construct the entire \codebook:
\begin{equation}
\bm{E}^{*} = \prod_{c=1}^{C} \bm{E}^{c}  \in \mathbb{R}^{(\prod_{c}M_c) \times (Cd)},
\end{equation}
representing all combinations of interest from each \scb.

For an item with embedding \(\bm{v}_i\), the interest encoder \(z_\text{enc}(\cdot)\) first transforms it to a latent representation \(\bm{r}_1 = z_\text{enc}(\text{sg}[\bm{v}_i]) \in \mathbb{R}^d\), which is then mapped to the nearest interest embedding \(\bm{e}^{1}_{m^1}\) in the first \scb $\bm{E}^{1}$, corresponding to the \(m^1\)-th row of \(\bm{E}^{1}\). 
Next, we compute the residual \(\bm{r}_c = \bm{r}_{c-1} - \text{sg}[\bm{e}^{c-1}_{m^{c-1}}]\) and repeat this process for \(C\) iterations:
\begin{equation}\label{eq: interest-vec}
z_\text{quan}^c(\bm{v}) = \bm{e}^c_{m^c}, \quad
\text{where} \quad m^c = \arg \min_j \|\bm{r}_c - \bm{e}^c_j\|_2^2.
\end{equation}
Here, we leverage the stop-gradient operator \(\text{sg}[\cdot]\), ensuring that \(\bm{r}\)'s gradients update only the interest encoder \(z_\text{enc}(\cdot)\), but not the item embedding \(\bm{v}_i\) or interest embedding \(\bm{e}\).

The multi-dimensional interest embedding $z_\text{quan}(\bm{v}_i)$ is obtained by concatenating the embeddings  \(\bm{e}^c_{m^c}\) from each \scb:
\begin{equation} \label{eq: interest-vec2}
z_\text{quan}(\bm{v}_i) = \bm{e}_{m^*} = \text{concat}(\bm{e}^{1}_{m^1}, \bm{e}^{2}_{m^2}, \dots, \bm{e}^{C}_{m^C}).
\end{equation}

Subsequently,  \(z_\text{quan}(\bm{v}_i)\) is processed by a decoder $z_\text{dec}(\cdot)$ to reconstruct the input \(\bm{v}_i\).
The quantization loss consists of three terms: the reconstruction loss for the decoder, the embedding loss for the codebook, and the commitment loss for the encoder:
\begin{equation}
\begin{small}
\begin{aligned}
\mathcal{L}_{\text{\cs}} =\ \  &\|\text{sg}[\bm{v}]-z_\text{dec}(z_\text{quan}(\bm{v}))\|_2^2 \\ 
 & +\sum_{c=1}^{C} \left( \|\text{sg}[\bm{r}_c] - \bm{e}_{m^c}^{c}\|_2^2 + \beta \|\bm{r}_c - \text{sg}[\bm{e}_{m^c}^{c}]\|_2^2 \right),
\end{aligned}
\end{small}
\end{equation}
where $\beta$ is to balance the learning objectives of reconstruction and commitment for the encoder. 
The latter two terms ensure that the interest embedding aligns with the output of interest encoder.

\subsection{\distribution}
In this section, we introduce an independent next-interest prediction model decoupled from the user tower for interest generation.

In \ds{}, we utilize a user-conditioned Generative Pre-Trained Transformer (GPT) to learn users' sequential behaviors conditioned on some key user features $\bm{u}_\text{feat}$(gender, age, etc), whose effectiveness has been proven in prior researches~\cite{feng2022recommender, rajput2024recommender}. 
The behavior sequence \(\mathcal{I}_u = \{i_u^1, i_u^2, \ldots, i_u^t\}\) is first converted to item embedding sequence \(\mathcal{S}_u = \{s_u^1, s_u^2, \ldots, s_u^t\}\) through embedding lookup operation, and then combined with encoded key user features \(z_\text{feat}(\bm{u}_\text{feat})\), namely the "user condition'', to construct the input to GPT:
\begin{equation}
\mathcal{S}_u^{\prime} = \text{concat}(z_\text{feat}(\bm{u}_\text{feat}), \mathcal{S}_u).
\end{equation}

Next, \(\mathcal{S}_u^{\prime}\) is fed into multiple GPT blocks, and the output at the last position of the final GPT block, denoted as \(z_\text{GPT}(\mathcal{S}_u^{\prime})\), is further passed through a fully connected layer \(z_\text{FC}\) to adapt it for the task of interest generation. 
Thus, a probability distribution over multi-dimensional interests in \cs{} is obtained:
\begin{equation}
\hat{\bm{p}_u} = \text{softmax}(z_\text{FC}(z_\text{GPT}(\mathcal{S}_u^{\prime}))),
\end{equation}
where \(z_\text{FC}(z_\text{GPT}(\mathcal{S}_u^{\prime})) \in \mathbb{R}^{\prod_{c}M_c}\) represents the predicted logits for all multi-dimensional interests.

To obtain the multi-dimensional interest label at the next time step, we pass the ground truth next time item \(i_u^{t+1}\) through \cs{} described in Equation~\ref{eq: interest-vec}. 
Assume that \(i_u^{t+1}\) belongs to the \(m^{1}\)-th interest in the first \scb, the \(m^{2}\)-th in the second, and so on. 
Then the corresponding multi-dimensional interest index \(m^*\) in the entire \codebook{} is computed as:
\begin{align} \label{m-transfer}
m^* = \sum_{c=1}^{C} \left( m^{c} \cdot \prod_{j=1}^{c-1} M_j \right),
\end{align}
where \(M_j\) is the total number of the interests in the \scb~ \(\bm{E}^{j}\).
The corresponding one-hot label \( \bm{p}_u \in \{0, 1\}^{\prod_{c}M_c}\) is constructed by setting the \(m^*\)-th index to \(1\), while others to \(0\).

Finally, the \ds{} loss is formulated via cross-entropy:
\begin{equation}
\mathcal{L}_{\text{\ds}} = - \bm{p}_u^\top \log \hat{\bm{p}_u}.
\end{equation}

\subsection{\prediction}
In this section, we describe how \ps{}
retrieve items using the interest embedding \(\bm{e}_{m^*}\), user embedding \(\bm{u}\) from the user tower, and item embedding \(\bm{v}_i\) from the item tower.

During training, the interest index $m^*$ is extracted from the ground truth next item according to Equation~\ref{eq: interest-vec}. During inference, it is sampled from the posterior distribution \(\hat{\bm{p}_u}\) generated by \ds. 
Once $m^*$ is obtained, the interest embedding \(\bm{e}_{m^*}\) can be retrieved from the \codebook~ according to Equation~\ref{eq: interest-vec2}.

We combine the interest embedding \(\bm{e}_{m^*}\) with the user embedding \(\bm{u}\) and pass the result through a fusion network \(z_\text{fusion}(\cdot)\).
This fused user-interest representation is then used to compute the preference score by a dot product with the item embedding \(\bm{v}_i\):
\begin{equation}
\hat{y}(\bm{e}_{m^*}, \bm{u}, \bm{v}_i) = z_\text{fusion}(\text{concat}(\bm{e}_{m^*}, \bm{u}))^\top \bm{v_i}.
\end{equation}

For training the retrieval module on interest-user-item tuples \((\bm{e}_{m^*}, \bm{u}, \bm{v}_i)\), we adopt the classical in-batch negative sampling~\cite{yi2019sampling} to select negative items, denoted as:
\begin{equation}
\mathcal{N}_i^- = \mathcal{B} \backslash \{i\},
\end{equation}
where \(\mathcal{B}\) includes all items in the current batch.

Meanwhile, to distinguish between the ground truth interest \(m^*\) and non-relevant interests,  negative interests are selected as: 
\begin{equation}
\mathcal{N}_m^- = (\tilde{\mathcal{M}} \cup \overline{\mathcal{M}})\backslash \{m^*\},
\end{equation}
where \(\tilde{\mathcal{M}}\) denotes some top interests from \(\hat{\bm{p}_u}\) predicted by \ds{}, serving as hard negatives,
and \(\overline{\mathcal{M}}\) are easy negative interests randomly sampled from the entire \codebook{}.

Finally, we apply the softmax loss to distinguish the positive sample \((\bm{e}_{m^*}, \bm{u}, \bm{v}_i)\) from negative items and negative interests:
\begin{equation}
\footnotesize
\mathcal{L}_{\ps} = -\log \left( \frac{e^{\hat{y}(\bm{e}_{m^*}, \bm{u}, \bm{v}_i)}}{e^{\hat{y}(\bm{e}_{m^*}, \bm{u}, \bm{v}_i)} + \sum_{i^\text{-} \in \mathcal{N}_i^-} e^{\hat{y}(\bm{e}_{m^*}, \bm{u}, \bm{v}_{i^\text{-}})} + \sum_{m^\text{-} \in \mathcal{N}_m^-} e^{\hat{y}(\bm{e}_{m^\text{-}}, \bm{u}, \bm{v}_i)}} \right)
\end{equation}

\subsection{Optimization}
\subsubsection{Training Objectives}
The overall loss function for training \rec~ is a weighted sum of losses from three modules:
\begin{equation}
\mathcal{L}_{\text{total}} = \lambda_1 \mathcal{L}_{\text{\cs}} + \lambda_2 \mathcal{L}_{\text{\ds}} + \lambda_3 \mathcal{L}_{\text{\ps}},
\end{equation}
where $\lambda_1$, $\lambda_2$ and $\lambda_3$ are hyper-parameters. 

\subsubsection{Joint \codebook{} Updates}
\label{Interest Dictionary Updates}
\paragraph{Update Rule}
We apply a joint training for the \codebook{} through \cs~and \ps{} to better align its semantic separation with downstream recommendation,
with the following update rule:
\begin{equation}
\bm{E}_{\text{(new)}} = \bm{E}_{\text{(old)}} - \eta (\lambda_1 \frac{\partial L_{\cs}}{\partial \bm{E}} + \lambda_3 \frac{\partial L_{\ps}}{\partial \bm{E}})
\end{equation}
where \( \frac{\partial (\cdot)}{\partial (\cdot)} \) represents the gradient, and \(\eta\) is the learning rate.

\paragraph{Training Strategy}
To ensure training stability of simultaneous updates, we adopt a three-stage strategy.
First, we train \ps~ independently while keeping \cs~ frozen and the \codebook{} unchanged.
Next, we train \cs~separately until the \codebook{} converges.
Finally, both \ps~and \cs~ are trained jointly, allowing simultaneous updates to the \codebook.

\paragraph{Initialization}
To improve the utilization of \codebook{}, following work~\cite{zeghidour2021soundstream}, we adopt a clustering-based initialization for each \scb{}, performing k-means clustering on the first training batch, and setting the resulting centroids as initial interest embeddings.
Additionally, the first \scbs~is initialized via preset prior categories (e.g., Music, Food, Education, etc.) to accelerate convergence.

Algorithm~\ref{alg:training} depicts the training phase of \rec.

\subsection{Online Serving}
\subsubsection{User Top-$K$ Interest Cache}\label{sec: user-topk}
We design a user interest cache (Figure~\ref{fig: overall}) to support efficient online serving.
This cache stores quantized top-$K$ interest indices generated by \ds~ with shape $K \times C$ for each user during online streaming training, allowing \cs{} and \ds{} to employ models of arbitrary design without increasing online inference latency. 

To provide flexibility for online diversity, we introduce the \textit{controllable aggregation} to control the diversity among the top-$K$ selected interests \(m_{1}^*, \cdots, m_{K}^*\) for each user. 
Specifically, for each \(m_{k}^{*}\), its corresponding index in each \scb{} is denoted as \(m_{k}^1, \cdots, m_{k}^C\).  
In this way, the number of times each \(m_{k}^c\) appears in the top-$K$  cache for an individual user is limited to a threshold \( \varepsilon \in \mathbb{Z}^{+}\). 
Formally, we enforce the following constraint:
\begin{equation}
\forall k \in [1, K], \ c \in [1, C], \quad \text{count}(m_k^c, \{m_1^c, \dots, m_K^c\}) \leq \varepsilon.
\end{equation}

Overall, each user's cached top-$K$ interest indices are updated in real-time, employing the above techniques.

\subsubsection{Online Recommendation}
The $K \times C$ interest indices \(m_{1,\cdots, K}^{1,\cdots, C}\) stored in the user top-$K$ cache are used to look up the corresponding interest embeddings \(\bm{e}_{m_{1, \cdots, K}^*}\) in the \codebook{}. 
These interest embeddings are fed into \ps{}, where each interest generates a corresponding user-interest representation \(\bm{u}_k\), based on the user tower output \(\bm{u}\).
After obtaining multiple user-interest representations \(\bm{u}_{1,\cdots,K}\), we request the Approximate Nearest Neighbor (ANN) service for the final recommendations.

Algorithm~\ref{alg:inference} details the inference phase of \rec.

\begin{algorithm}[t]
    \caption{Streaming Training of \rec}  
    \label{alg:training}  
    \begin{algorithmic}[1]
        \State \textbf{Initialize} parameters for \cs, \ds, and \ps.
        \State \textbf{Train} \ps~separately while fixing the \codebook.
        \State \textbf{Train} \cs~separately and optimize the \codebook.
        \While {system is running}  \Comment{Online Learning}
            \State Observe: Real-time user interaction \((u, i)\).
            \State Obtain the quantified interest \(m^*\) for item i through \cs.
            \State Sample negative items \( \mathcal{N}^-_{i} \) and negative interests \( \mathcal{N}^-_{m} \).
            \State \textbf{Update} \cs~ using loss \( \mathcal{L}_{\cs}(i) \).
            \State \textbf{Update} \ds~ using loss \( \mathcal{L}_{\ds}(u, m^*) \).
            \State \textbf{Update} the top-K interest cache for user u.
            \State \textbf{Update} \ps~ using loss \( \mathcal{L}_{\ps}(u, i, m^*, \mathcal{N}^-_{i}, \mathcal{N}^-_{m}) \).
        \EndWhile
    \end{algorithmic}  
\end{algorithm}

\begin{algorithm}[t]
    \caption{Inference of \rec}  
    \label{alg:inference}
    \begin{algorithmic}[1]
        \Require User embedding \( \bm{u} \), cached user interest indices \(m_{1,\cdots, K}^{1,\cdots, C}\).
        \Require Item embeddings \( \bm{v}_i \) for all candidate items.
        \Ensure Recommendation list of \( N \) items.
        
        \State Fetch interest embeddings \( \bm{e}_{m^*_1, \cdots, m^*_K} \) from \codebook:
        \[
            \bm{e}_{m^*_k} = \text{concat}(\bm{e}^{1}_{m^1_k}, \bm{e}^{2}_{m^2_k}, \dots, \bm{e}^{C}_{m^C_k}).
        \]
        \State Compute fused representation \( \bm{u}_{1, \dots, K} \):
        \[
        \bm{u}_k = z_{\text{fusion}}(\text{concat}(\bm{e}_{m_k^*}, \bm{u})), \quad k = 1, 2, \dots, K.
        \]
        \State Retrieve items with respect to each $\bm{u}_k$:
        \[
            \hat{y}_k(\bm{v}_i) = (\bm{u}_k)^\top \bm{v}_i.
        \]
        \State \Return Retrieved Top-\( N \) recommendation set.
    \end{algorithmic}  
\end{algorithm}

\section{Theoretical Analysis} \label{sec: Theoretical Analysis}
In this section, we provide analyses for the discussion in section~\ref{sec: introduction}. We first present the proposition that the interest quantization induces a Voronoi partition~\cite{de2008computational} (Proposition \ref{proposition: Induce}). Then it follows that the quantization offers a non-trivial lower bound of separation (Corollary \ref{corollary: Amplified Separation}) which can propagate to the retrieval space under mild conditions (Proposition \ref{proposition:retrieval separability}), while regularization offers no such lower-bound guarantee (Proposition \ref{proposition: Regularization}).
\begin{proposition}[Interest Quantization Induces a Voronoi Partition]\label{proposition: Induce}
The \codebook{} $E^*=\{e_1,\dots,e_{|E^*|}\}\subset\mathbb{R}^{Cd}$ induces the Voronoi partition.
\end{proposition}

\begin{proof}
See Appendix \ref{appendix: proof} for details.
\end{proof}

\begin{corollary}[Structural Separation Induced by Discrete Indices]
If two data points are quantized into different Voronoi cells $V_{m}$ and $V_{n}$, then since the \codebook{} $E$ is finite, there exists a strictly positive minimum distance
\[
\Delta_{\min}=\min_{m\neq n}\|\mathbf{e}_{m}-\mathbf{e}_{n}\|_2 > 0,
\]
which provides a non-trivial lower bound of separation between any two distinct Voronoi cells in the codebook space. \qed
\end{corollary}

\begin{corollary}[Amplified Separation in \codebook{}]\label{corollary: Amplified Separation}
Let $\delta_c$ denote the minimum pairwise distance in \scb{} $E^{c}$.  
If two interest embeddings $\mathbf{e}_{m},\mathbf{e}_{n}\in E^*$ differ in $h$ indices, then
\[
\|\mathbf{e}_{m}-\mathbf{e}_{n}\|_2^2 \;\ge\; \sum_{c} \delta_c^2 \;\ge\; C \cdot \Big(\min_c \delta_c\Big)^2,
\]
Thus, the separation between distinct cells increases with the Hamming distance between their index tuples. \qed
\end{corollary}

\begin{proposition}[From \codebook{} separability to retrieval separability]
\label{proposition:retrieval separability}
Assume the fusion function $z_{\text{fusion}}(u, e)$ satisfies local
Lipschitz-type condition in its second argument. 
If two interests are separated by at least $\Delta$, then there exists a constant $0<\alpha \le \infty$ such that their fused retrieval vectors satisfy
\[
\|z_{\text{fusion}}(u,e_m) - z_{\text{fusion}}(u,e_n)\|_2 \;\ge\; \alpha \,\Delta .
\]
In cosine-similarity maximum inner product search (MIPS), this further implies a strictly 
positive score-margin lower bound, which in turn upper-bounds the overlap between the 
Top-$N$ candidate sets of the two interests.
\end{proposition}

\begin{proof}
See Appendix \ref{appendix: proof} for details.
\end{proof}

\begin{proposition}
(Regularization does not imply structural separation.) \label{proposition: Regularization} 
For any finite regularization weight $\lambda$ and any continuous penalty $\mathcal{R}$, there does not exist a 
data-independent constant $\Delta > 0$ such that all global minimizers satisfy
\[
\min_{k \neq \ell}\|\mathbf{u}_k - \mathbf{u}_\ell\|_2 \;\ge\; \Delta.
\]
\end{proposition}

\begin{proof}
See Appendix \ref{appendix: proof} for details.
\end{proof}

\begin{table}[t]
\centering
\renewcommand{\arraystretch}{0.9}
\caption{Statistics of datasets.}
\scalebox{0.9}{\begin{tabular}{l|rrr}
\toprule
\textbf{Dataset} & \textbf{\# users} & \textbf{\# items} & \textbf{\# interactions} \\
\hline
Amazon Books  & 603,668  & 367,982   & 8,898,041  \\
\hline
\makecell{Amazon Clothing,\\ Shoes and Jewelry}  & 1,219,678  & 376,858   & 11,285,464  \\
\hline
RetailRocket  & 33,708  & 81,635   & 356,840  \\
\hline
Rednote        &  238,960,609 & 91,784,192 & 26,989,379,219  \\
\bottomrule
\end{tabular}}
\label{table: dataset Statistics}
\end{table}
\begin{table*}[htbp]
\footnotesize
\caption{Experimental results comparing the performance of \rec{} with baseline methods, across various metrics.}
\centering
\renewcommand{\arraystretch}{0.9}
\begin{tabular}{l|l|ccccccccccc|c}
\toprule
\textbf{Dataset} & \textbf{Metric} & \textbf{MIND} & \textbf{ComiRec} & \textbf{RE4} & \textbf{UMI} & \textbf{MGNM} & \textbf{TiMiRec} & \textbf{SINE} & \textbf{REMI} & \textbf{SimRec} & \textbf{DisMIR} & \textbf{\rec{}} & \textbf{Improve.} \\
\midrule

\multirow{6}{*}{\textbf{Amazon Books}} 
    & Recall@20 & 0.0438 & 0.0543 & 0.0599 & 0.0701 & 0.0727 & 0.0780 & 0.0822 & 0.0840 & 0.0852 & \uline{0.0879} & \textbf{0.0977} & 11.15\% \\
    & HR@20 & 0.0906 & 0.1109 & 0.1236 & 0.1416 & 0.1462 & 0.1589 & 0.1620 & 0.1676 & 0.1714 & \uline{0.1804} & \textbf{0.1932} & 8.08\% \\
    & NDCG@20 & 0.0339 & 0.0408 & 0.0463 & 0.0527 & 0.0543 & 0.0586 & 0.0618 & 0.0625 & 0.0640 & 0.0674 & \textbf{0.0743} & 10.24\% \\
    & Recall@50 & 0.0679 & 0.0850 & 0.0993 & 0.1052 & 0.1086 & 0.1143 & 0.1184 & 0.1197 & 0.1248 & \uline{0.1368} & \textbf{0.1541} & 12.62\% \\
    & HR@50 & 0.1380 & 0.1723 & 0.1979 & 0.2067 & 0.2133 & 0.2205 & 0.2246 & 0.2324 & 0.2413 & \uline{0.2634} & \textbf{0.2774} & 5.32\% \\
    & NDCG@50 & 0.0397 & 0.0480 & 0.0572 & 0.0594 & 0.0616 & 0.0634 & 0.0664 & 0.0668 & 0.0693 & \uline{0.0752} & \textbf{0.0821} & 9.16\% \\
\midrule

\multirow{6}{*}{\makecell{\textbf{Amazon }\\ \textbf{Clothing, Shoes}\\ \textbf{ and Jewelry}}} 
    & Recall@20 & 0.0343 & 0.0464 & 0.0496 & 0.0621 & 0.0655 & 0.0702 & 0.0722 & 0.0763 & 0.0774 & \uline{0.0800} & \textbf{0.0942} & 17.75\% \\
    & HR@20 & 0.0793 & 0.1011 & 0.1143 & 0.1286 & 0.1340 & 0.1459 & 0.1514 & 0.1539 & 0.1582 & \uline{0.1682} & \textbf{0.1787} & 6.24\% \\
    & NDCG@20 & 0.0306 & 0.0328 & 0.0425 & 0.0480 & 0.0499 & 0.0534 & 0.0554 & 0.0567 & 0.0594 & \uline{0.0659} & \textbf{0.0760} & 15.31\% \\
    & Recall@50 & 0.0628 & 0.0816 & 0.0943 & 0.1010 & 0.1045 & 0.1114 & 0.1127 & 0.1173 & 0.1215 & \uline{0.1314} & \textbf{0.1429} & 8.74\% \\
    & HR@50 & 0.1253 & 0.1647 & 0.1903 & 0.1967 & 0.2034 & 0.2117 & 0.2223 & 0.2284 & 0.2369 & \uline{0.2555} & \textbf{0.2698} & 5.59\% \\
    & NDCG@50 & 0.0378 & 0.0429 & 0.0510 & 0.0519 & 0.0544 & 0.0604 & 0.0625 & 0.0639 & 0.0669 & \uline{0.0738} & \textbf{0.0805} & 9.09\% \\
\midrule

\multirow{6}{*}{\makecell{\textbf{RetailRocket}}} 
    & Recall@20      & 0.0991 & 0.1035 & 0.1397 & 0.1519 & 0.1664 & 0.1958 & 0.2085 & 0.2129 & 0.2206 & \uline{0.2385} & \textbf{0.2637} & 10.57\% \\
    & HR@20    & 0.1429 & 0.1602 & 0.2103 & 0.2364 & 0.2585 & 0.2912 & 0.3078 & 0.3183 & 0.3285 & \uline{0.3524} & \textbf{0.3715} & 5.42\% \\
    & NDCG@20        & 0.0570 & 0.0609 & 0.0785 & 0.0875 & 0.0901 & 0.1033 & 0.1105 & 0.1198 & 0.1237 & \uline{0.1330} & \textbf{0.1501} & 12.86\% \\
    & Recall@50      & 0.1597 & 0.1666 & 0.2194 & 0.2423 & 0.2646 & 0.2928 & 0.3105 & 0.3160 & 0.3246 & \uline{0.3447} & \textbf{0.3845} & 11.55\% \\
    & HR@50    & 0.2464 & 0.2501 & 0.3174 & 0.3574 & 0.3786 & 0.4153 & 0.4377 & 0.4515 & 0.4605 & \uline{0.4815} & \textbf{0.5286} & 9.78\% \\
    & NDCG@50        & 0.0634 & 0.0684 & 0.0884 & 0.0974 & 0.1033 & 0.1167 & 0.1212 & 0.1281 & 0.1305 & \uline{0.1360} & \textbf{0.1551} & 14.04\% \\
\midrule

\multirow{6}{*}{\textbf{Rednote}} 
    & Recall@120 & 0.0588 & 0.0722 & 0.0771 & 0.0854 & 0.0890 & 0.0956 & 0.0970 & 0.1024 & 0.1042 & \uline{0.1083} & \textbf{0.1395} & 28.78\% \\
    & HR@120 & 0.1114 & 0.1348 & 0.1487 & 0.1623 & 0.1682 & 0.1803 & 0.1914 & 0.1956 & 0.1995 & \uline{0.2088} & \textbf{0.2330} & 11.58\% \\
    & NDCG@120 & 0.0517 & 0.0548 & 0.0618 & 0.0722 & 0.0746 & 0.0787 & 0.0811 & 0.0824 & 0.0839 & \uline{0.0873} & \textbf{0.1180} & 35.16\% \\
    & Recall@200 & 0.0850 & 0.1043 & 0.0886 & 0.1273 & 0.1311 & 0.1352 & 0.1367 & 0.1438 & 0.1498 & \uline{0.1639} & \textbf{0.2295} & 40.00\% \\
    & HR@200 & 0.1629 & 0.1982 & 0.2278 & 0.2371 & 0.2436 & 0.2529 & 0.2654 & 0.2684 & 0.2771 & \uline{0.2993} & \textbf{0.3064} & 2.38\% \\
    & NDCG@200 & 0.0578 & 0.0647 & 0.0734 & 0.0770 & 0.0791 & 0.0818 & 0.0846 & 0.0861 & 0.0895 & \uline{0.0975} & \textbf{0.1257} & 28.96\% \\
\bottomrule
\end{tabular}
\label{Table: overallResults}
\end{table*}

\section{Experiments}
To evaluate our method, we conduct experiments to answer the following research questions (RQs):
\begin{itemize}[leftmargin=*]
    \item \textbf{RQ1:} How does \rec{} perform compared to baselines?
    \item \textbf{RQ2:} How does \codebook{} look like and functions?
    \item \textbf{RQ3:} How does \rec{} perform in interest collapse (3-1) and interest evolution modeling (3-2)?
    \item \textbf{RQ4:} How do the ablations of the framework design (4-1) and the optimization technique (4-2) perform?
    \item \textbf{RQ5:} How do hyperparameters affect overall performance (5-1), interest collapse (5-2), and interest evolution modeling (5-3)?
    \item \textbf{RQ6:} How does \rec{} perform in real-world production?
\end{itemize}
\subsection{Experimental Settings}
\subsubsection{Datasets.}
We conduct experiments on four real-world datasets, with their statistics presented in Table~\ref{table: dataset Statistics}.

\begin{itemize}[leftmargin=*]
\item \textbf{Amazon}~\cite{ni2019justifying}: This dataset is derived from the Amazon Review Dataset. 
Following prior studies \cite{cen2020controllable, zhang2022re4, du2024disentangled}, we choose the 5-core subset "Book" and "Clothing, Shoes, and Jewelry" for evaluation.
\item \textbf{RetailRocket}~\cite{roman_zykov_noskov_artem_anokhin_alexander_2022}: This dataset is collected from a real-world e-commerce website. We treat views as implicit feedback and filter out users and items with less than 5 records.
\item \textbf{Rednote}: A large-scale industry dataset collected from a real-world content-sharing platform, Rednote (Xiaohongshu) for offline evaluation, containing user-view interactions with notes.
\end{itemize}

We split each user's interaction chronologically into 80\% for training, 10\% for validation, and 10\% for testing. The maximum sequence length is set to 100 for industry dataset and 20 for others.

\subsubsection{Baselines.}
We include MIND~\cite{li2019multi}, 
ComiRec~\cite{cen2020controllable}, RE4~\cite{zhang2022re4}, UMI~\cite{chai2022user}, 
MGNM~\cite{tian2022multi}, SINE~\cite{tan2021sparse}, TiMiRec \cite{wang2022target},
REMI~\cite{xie2023rethinking}, SimRec~\cite{liu2024attribute},
and DisMIR~\cite{du2024disentangled} for comparison.



\subsubsection{Implementation Settings.} 
All methods are optimized by Adam optimizer with $lr = 0.001$. 
For \rec{}, we use LeakyReLU except for MIPDM, which uses GELU.
In \cs{}, the encoder and decoder use MLPs, with hidden layers [256, 128, 64, 16] and [64, 128, 256], respectively.
The \codebook{} has $4$ \scbs{} with $M_{1,2,3,4}=32,16,8,4$ and $d=16$. 
\ds{} uses a $6$-layer GPT model, with $4$ attention heads, hidden size $16$. 
The loss weights $\beta$, $\lambda_1$, $\lambda_2$, and $\lambda_3$ are set to $0.25$, $0.2$, $1$, and $1$, respectively. 
In \ps{}, both the user tower and item tower use MLPs, with hidden layers [1024, 512, 256, 64] and [512, 512, 128, 64], respectively. 
The interest fusion network $z_{\text{fusion}}$ has hidden layers [256, 64].
The hyper-parameters for the user top-$K$ cache are set as $\varepsilon=3$, and $K=5$.
For baseline models, we use original hyperparameters whenever available.
Otherwise, we tune them for optimal performance.

\subsubsection{Evaluation Metrics} \label{sec: metrics}

\paragraph{Overall Metrics}\hspace{0.2em}
Following prior research~\cite{cen2020controllable}, we use Recall@N, HR@N, and NDCG@N to evaluate recommendation performance.




\tightparagraph{Metrics Parameter.}
We set $N$ to $20$ and $50$ for the first three datasets following \cite{cen2020controllable}.  
For the large-scale industry dataset \textit{Rednote}, Metric@$20$ and Metric@$50$ are too narrow for meaningful evaluation. 
Therefore, we use Metric@$120$ and Metric@$200$.

\subsection{Overall Performance (RQ1)}
We summarize the following findings from the overall performance in Table~\ref{Table: overallResults}.
To begin with, recent multi-interest methods such as DisMIR, SimRec, REMI, SINE, and TiMiRec significantly outperform earlier methods like MIND and ComiRec. This suggests the potential benefits of their efforts in mitigating interest collapse and improving interest evolution modeling. It is further supported by the metrics AMR@N and CUR@N reported in the latter Section.
Moreover, GemiRec achieves the best overall performance, especially pronounced in the large-scale Industry dataset. In this setting, where user interests are highly diverse and continuously evolving, addressing interest collapse and evolution plays a more crucial role in improving overall performance.
It indicates that the proposed interest quantization and generation framework is more effective at tackling the aforementioned challenges, which is further validated in Sections \ref{Interest Collapse and Evolution (RQ3)}.
In summary, \rec{} outperforms baselines, demonstrating its superior overall performance.


\begin{figure}
    \centering
    \begin{subfigure}{0.225\textwidth} 
        \centering
        \includegraphics[width=\linewidth]{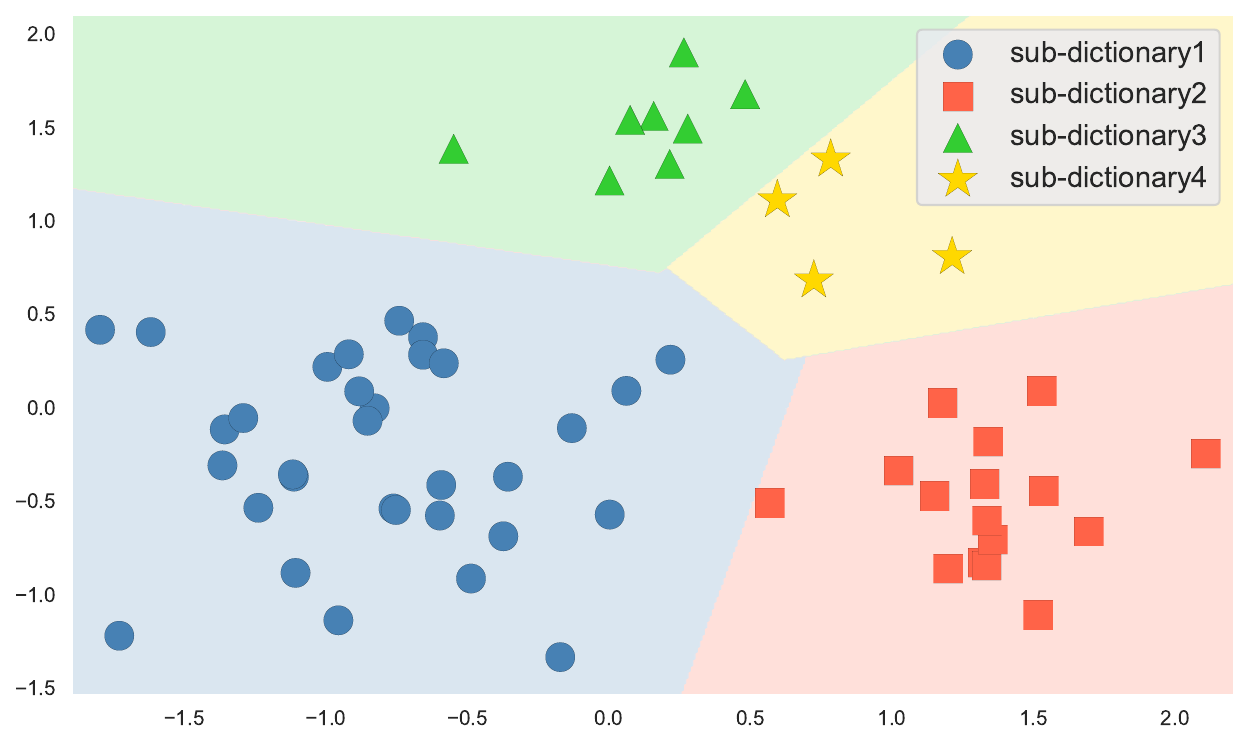} 
        \caption{\normalfont \scriptsize t-SNE visualization of the learned interest embeddings in the \codebook.} 
    \label{fig: independence}
    \end{subfigure}
    \hspace{0.0\textwidth} 
    \begin{subfigure}{0.225\textwidth} 
        \centering
        \includegraphics[width=\linewidth]{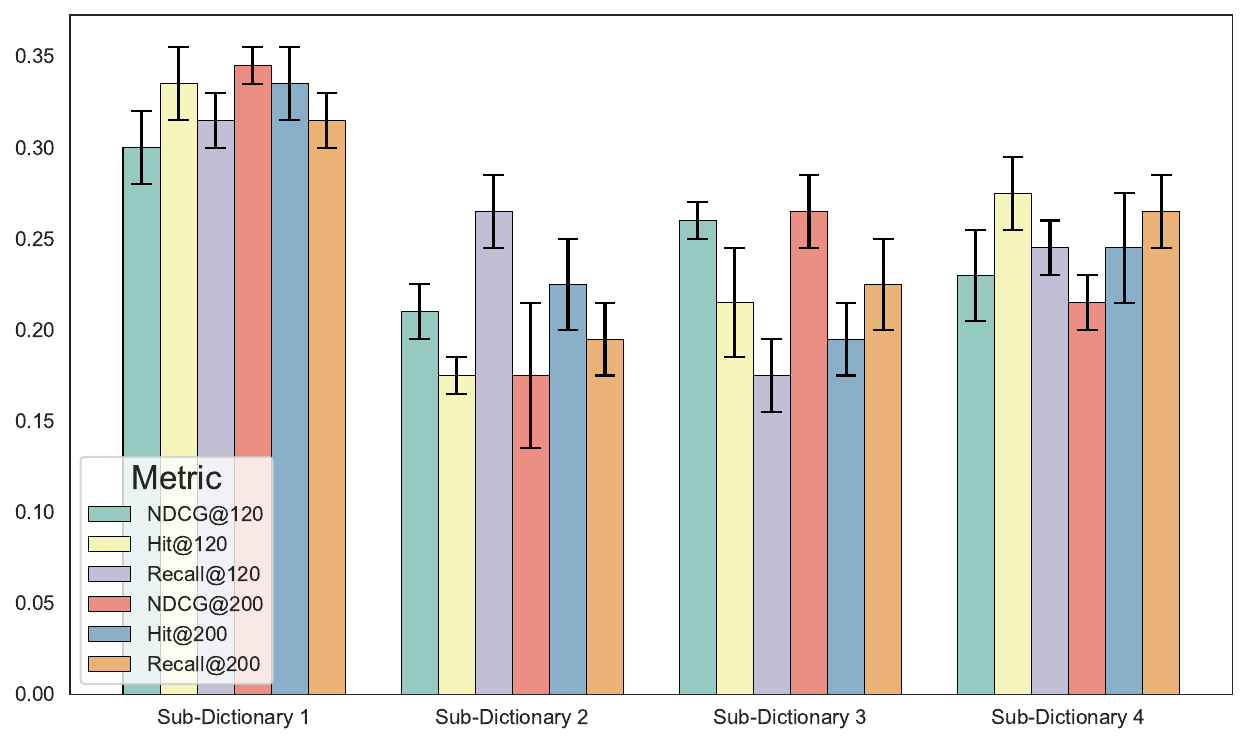} 
        \caption{\normalfont \scriptsize Frequency of each sub-dictionary acting as the most effective for recommendation.} 
    \label{fig: importance}
    \end{subfigure}
    \caption{Distribution and importance of \codebook.} 
    \label{fig: all 3 images}
\end{figure}

\subsection{Analysis of the \codebook{} (RQ2)}
\subsubsection{Interest Embeddings} \label{sec: independence-btw-cb}
To analyze the distribution of learned \codebook{}, we employ t-SNE~\cite{van2008visualizing} visualization on the industry dataset. 
Figure~\ref{fig: independence} shows that interests from different \scbs{} lie in distinct regions and exhibit uniform distribution patterns, indicating that each \scb~ captures distinct dimensions of user interests and learns representative embeddings.

\subsubsection{Importance of Sub-dictionaries}
We evaluate the necessity of each \scb{} by iteratively retaining one \scb{}’s interest indices in the top-$K$ predictions while randomizing the others, marking the highest metric scorer in each round as the winner. Results in Figure~\ref{fig: importance} show that the first \scb{} is most influential, yet others also contribute, validating the practical value of multi-dictionary design.
\subsubsection{Case study}
We present a case study on sampled sports news items, sharing the same indices---26 in the first sub-dictionary and 9 in the second. A closer observation shows that index 26 in the first sub-dictionary covers mainly sports, while index 9 in the second maps to trending events.


\subsection{Interest Collapse and Evolution (RQ3)} \label{Interest Collapse and Evolution (RQ3)}
\subsubsection{Analysis of Interest collapse (RQ3-1)} \label{Interest collapse}
We introduce a \textit{Alignment Margin Relevance@N (AMR@N)} to evaluate the semantic separation between user-interest representation $\bm{u}_{k}$ and retrieved candidate sets $R_u^{(k)}$, which is computed as:
\[
\mathrm{AMR@N} = \frac{1}{|U|K} \sum_{u,k} \frac{1}{|R_u^{(k)}|} \sum_{x \in R_u^{(k)}}  \big[ \cos(\bm{u}_{k}, \bm{v_x}) - \max_{j \ne k} \cos(\bm{u}_{j}, \bm{v_x}) \big].
\]

From the AMR@N in Table~\ref{Table:generalizationPerformance}, we can draw the following observations.
First, earlier multi-interest approaches like ComiRec~\cite{cen2020controllable}, do exhibit relatively severe interest collapse.
Secondly, recent methods such as 
REMI~\cite{xie2023rethinking}, SimRec~\cite{liu2024attribute}, and DisMIR~\cite{du2024disentangled} 
have achieved encouraging progress, but they still exhibit a degree of redundancy.
Lastly, \rec{} achieves the highest AMR scores, validating its effort in mitigating interest collapse, i.e., as discussed in Section \ref{sec: introduction},  the interest quantization
mechanism with the top-K indices selection.

\begin{table} [t!] 
\renewcommand{\arraystretch}{1.05}
\centering
 \caption{Performance in terms of interest collapse (AMR) and interest evolution modeling (CUR) on the Industry dataset.}
\scalebox{0.85}{
\begin{tabular}{l|cc|cc}
    \toprule
    \multicolumn{1}{l|}{\multirow{2}{*}{\textbf{Methods}}}  &  \multicolumn{4}{c}{\textbf{Industry}} \\
    \cline{2-5}
    & \multicolumn{1}{c}{AMR@120} & \multicolumn{1}{c|}{AMR@200} & \multicolumn{1}{c}{CUR@120} & \multicolumn{1}{c}{CUR@200} \\
    \hline
    
    \textbf{MIND}  
        & 0.0394 & 0.0402 & 0.0134 & 0.0184 \\
    \textbf{ComiRec}  
        & 0.0420 & 0.0415 & 0.0182 & 0.0204 \\
    
    \textbf{RE4}  
        & 0.0733 & 0.0719 & 0.0189 & 0.0227 \\

    \textbf{UMI}  
        & 0.0749 & 0.0721  & 0.0195 & 0.0217 \\

    \textbf{MGNM}  
        & 0.0779 & 0.0741 & 0.0375 & 0.0410 \\

    \textbf{SINE}  
        & 0.1448 & 0.1407 & 0.0191 & 0.0233 \\
    \textbf{TiMiRec}  
        & 0.0965 & 0.0927 & \underline{0.0460} & \underline{0.0649} \\
    
    \textbf{REMI}  
        & 0.1565 & 0.1487 & 0.0202 & 0.0259 \\

    \textbf{SimRec}  
        & 0.1572 & 0.1500 & 0.0198 & 0.0239 \\
    
    \textbf{DisMIR}  
        & \underline{0.1686} & \underline{0.1596} & 0.0194 & 0.0223 \\
    
    \hline
    \textbf{\rec{}} & \textbf{0.2104} & \textbf{0.2046} & \textbf{0.0842} & \textbf{0.1245} \\
    \bottomrule
\end{tabular}}
   \label{Table:generalizationPerformance}
\end{table}

\subsubsection{Analysis of Interest evolution (RQ3-2)} \label{Interest evolution}
We introduce a \textit{Category-Unseen Recall@N (CUR@N)}. It evaluates the recall ratio of user-interacted items from categories absent in the user's historical behavior sequence $\mathcal{I}_u$, which is computed as:
\[
\text{CUR@N} = \frac{1}{|U|} \sum_{u \in U} \frac{|R_u^N \cap J_u|}{|J_u|}
\]
where \(R_u^N\) is the top-N recommended item set for user \(u\), and \(J_u\) is the set of interacted items from categories absent in $\mathcal{I}_u$. 

From the CUR@N in Table~\ref{Table:generalizationPerformance}, we have the following findings: First, the capability of multi-interest methods in interest evolution modeling falls short. 
Secondly, MGNM~\cite{tian2022multi} and TiMiRec~\cite{wang2022target} have made great progress as expected.
Lastly, \rec{} outperforms baselines, demonstrating its effectiveness in interest evolution modeling, which can be attributed to, as discussed in Section~\ref{sec: introduction}, the design of interest generation, which is decoupled from the user tower and explicitly models evolving interests through an independent next-interest prediction task.

\subsubsection{Case study}
We present a case study on the relationship between learned user interests and interacted items. Specifically, we sample users from the industry dataset and project their interest representations $\mathbf{u}_{1,\cdots,K}$, along with embeddings of previously and future-interacted items using t-SNE~\cite{van2008visualizing}. Figure~\ref{fig: Interst collapse} suggests that \rec{} better captures diverse user interests and models their evolution than ComiRec. Similar advantages are observed over other methods, though we only present ComiRec here.

\begin{figure}
    \centering
    \includegraphics[width=\linewidth]{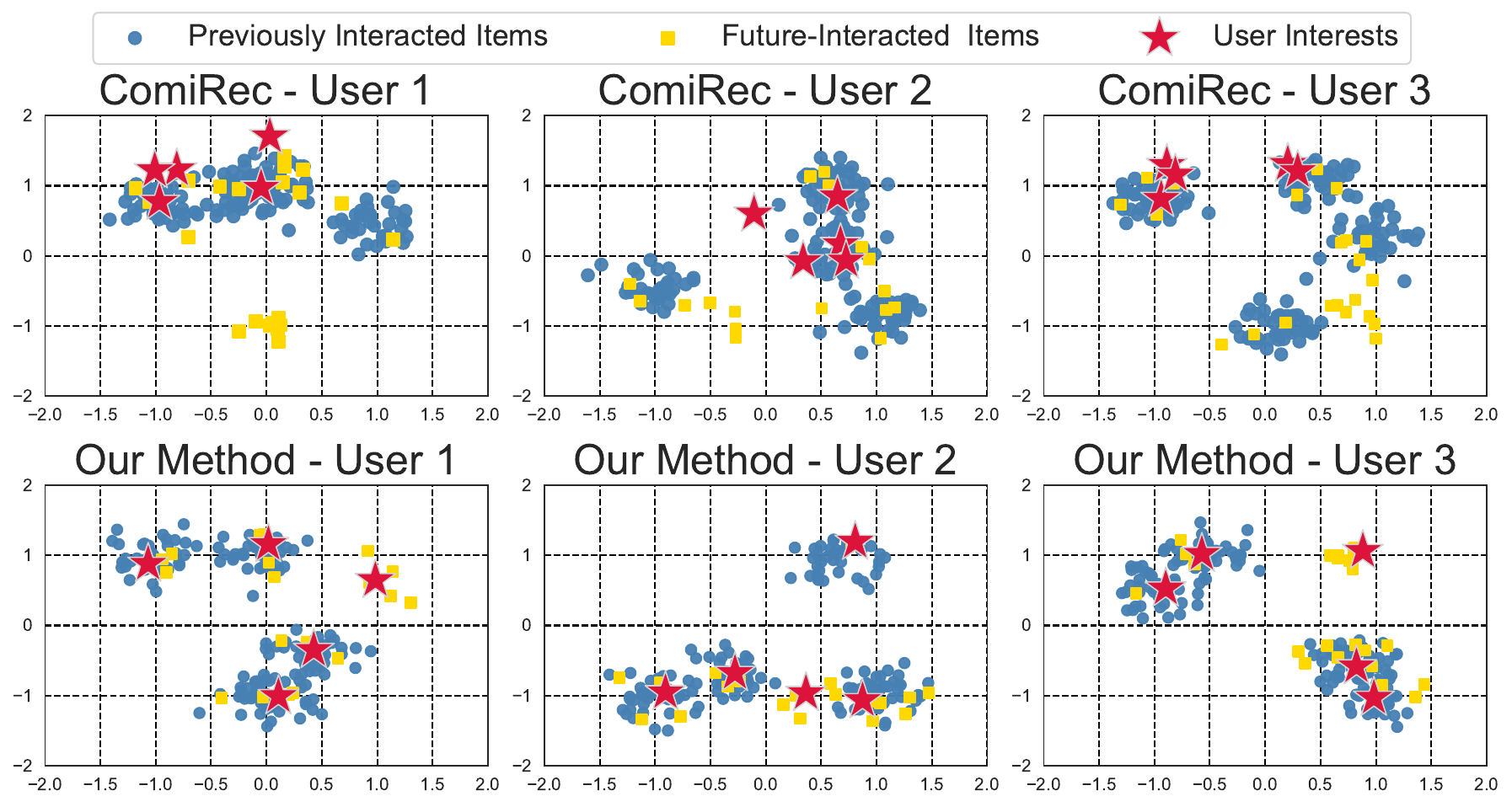}        \caption{\normalfont t-SNE visualization of user-interest representations, with embeddings of previously interacted and future-interacted items.} 
    \label{fig: Interst collapse}
\end{figure}
\subsection{Ablation Study (RQ4)}
\subsubsection{Module Ablation (RQ4-1)}
In this section, we conduct module ablation experiments on several variants of ~\rec{}, as follows:
\begin{itemize}[label=--, leftmargin=1em]
    \item \textbf{\rec{}-I}: Replaces the \codebook{} in \cs{} with preset categories as interest indices.
    \item \textbf{\rec{}-J}: Removes joint training between \cs{} and \ps{}.
    \item \textbf{\rec{}-U}: Removes user condition from the input in \ds{}.
    \item \textbf{\rec{}-M}: Replaces \ds{} by predicting top-$K$ interests based solely on their historical frequency.
\end{itemize}

Table~\ref{tab: module ablation} shows that the complete \rec{} consistently outperforms its variants, highlighting the importance of each component:
(1) The learned \codebook{} captures user interests more effectively than preset categories 
(2) Joint training between \cs{} and \ps{} enhances the adaptiveness of \codebook{} to downstream recommendation tasks.
(3) The user condition in ~\ds{} supplies crucial side information for accurate interest generation.
(4) \ds{} outperforming frequency-based predictions underscores the importance of explicitly modeling evolving user preferences.

Besides, the variants of \rec{} still outperform baselines, indicating that the primary gains originate from the proposed framework-level refinement, i.e, interest quantization and generation.

\begin{table}[t] 
\centering
\caption{Module Ablation Experiments.}
\scalebox{0.92}{
\begin{tabular}{|c|c|c|c|c|c|}
\hline
\textbf{Amazon-Book} & {\rec{}*} & {-I}& {-J} & {-U} & {-M} \\
\hline
Recall@20 & \textbf{0.0977} & 0.0911 &	0.0948	& 0.0915	& 0.0883 \\
\hline
Recall@50 & \textbf{0.1541} & 0.1380 &	0.1486	& 0.1403 &0.1356 \\
\hline
\textbf{Metric} & {\rec{}*} & {-I}& {-J} & {-U} & {-M} \\
\hline
Recall@120 & \textbf{0.1395} & 0.1169 & 0.1227 & 0.1196 & 0.1124 \\
\hline
Recall@200 & \textbf{0.2295} & 0.1845 & 0.2055 & 0.1805 & 0.1723 \\
\bottomrule
\end{tabular}}
\label{tab: module ablation}
\end{table}
\begin{table}[t] 
\caption{Ablation study of optimization techniques.
}
\centering
\scalebox{0.785}{
\begin{tabular}{|p{2.cm}|p{1.6cm}|p{1.6cm}|p{2.cm}|p{1.6cm}|}
\hline
\textbf{Metric} & \textbf{\rec{}} & \textbf{w/o 3-stage training} & \textbf{w/o kmeans initialization} & \textbf{w/o preset categories} \\
\hline
Converged Step & $\approx 500,000$ & $\text{NaN}$ & $\approx 450,000$ & $\approx 750,000$\\
\hline
\text{Utilization} & 93.3\% & -- & 21.6\% & 87.9\%
\\
\hline
Recall@120 & \textbf{0.1395} & -- & 0.1241 & 0.1306 \\
\hline
Recall@200 & \textbf{0.2295} & --  & 0.2079 & 0.2198 \\
\hline
\end{tabular}}
\label{tab: Ablation Study of Optimization Techniques.}
\end{table}

\subsubsection{Optimization Techniques Ablation(RQ4-2)}
We perform ablation studies on optimization techniques introduced in Section ~\ref{Interest Dictionary Updates}.  
The results shown in Table~\ref{tab: Ablation Study of Optimization Techniques.}, lead to the following findings:  
(1) The model fails to converge without the 3-stage training strategy.  
(2) The k-means initialization effectively improves the utilization of codebook from $21.6\%$ to $93.3\%$.  
(3) Initializing the first \scbs~by predefined categories not only accelerates convergence but also enhances performance by integrating external prior knowledge.

\begin{table}[t] 
\small
\centering
\caption{Hyperparameter experiments on Industry Dataset.}
\begin{subtable}{\linewidth}
    \centering
    \abovecaptionskip=2pt
    \begin{tabularx}{\linewidth}{|p{1.28cm}|p{0.62cm}|p{0.68cm}|p{0.95cm}|p{1.3cm}|X|}
    \hline
    \textbf{Metric} & \textbf{32} & \textbf{32-16} & \textbf{32-16-8} & \textbf{32-16-8-4*} & \textbf{32-16-8-4-4} \\
    \hline
    Recall@120 & 0.1169 & 0.1209 & 0.1254 & \textbf{0.1395} & 0.1303 \\
    \hline
    Recall@200 & 0.1845 & 0.1988 & 0.2075 & \textbf{0.2295} & 0.2236 \\
    \hline
    AMR@120 & 0.1925 & 0.2001 & 0.2045 & \textbf{0.2104} & 0.2090 \\
    \hline
    AMR@200 & 0.1861 & 0.1939 & 0.1976 & \textbf{0.2046} & 0.2023 \\
    \hline
    \end{tabularx}
    \caption{Dictionary sizes in \cs.}
    \label{tab:Performance of IDMM with Different Dictionary Sizes}
\end{subtable}
\begin{subtable}{\linewidth}
    \centering
    \abovecaptionskip=2pt
    \begin{tabularx}{\linewidth}{|X|X|X|X|X|}
    \hline
    \textbf{Metric} & \textbf{2 Layers} & \textbf{4 Layers} & \textbf{6 Layers*} & \textbf{8 Layers} \\
    \hline
    Recall@120 & 0.1180 & 0.1256 & 0.1395 & \textbf{0.1402} \\
    \hline
    Recall@200 & 0.2017 & 0.2108 &  0.2295 & \textbf{0.2320} \\
    \hline
    CUR@120 & 0.0756 & 0.0799 & 0.0842 & \textbf{0.0851} \\
    \hline
    CUR@200 & 0.1147 & 0.1184 &  0.1245 & \textbf{0.1263} \\
    \hline
    \end{tabularx}
    \caption{Number of GPT layers in \ds.}
    \label{tab:Effect of MIPDM Complexity}
\end{subtable}
\begin{subtable}{\linewidth}
\renewcommand{\arraystretch}{1.}
    \centering
    \abovecaptionskip=2pt
\begin{tabularx}{\linewidth}{|p{0.9cm}|X|p{1.65cm}|p{1.65cm}|}
\hline
\textbf{Model} & \textbf{Coeff ($\lambda_\text{1}/\lambda_{\text{2}}/\lambda_{\text{3}}$) }
& \textbf{Recall@120} & \textbf{Recall@200}\\ \hline
\text{default}         & 0.2/1/1              & \textbf{0.1395}        & \textbf{0.2295}         \\ \hline
\multirow{2}{*}{$\lambda_{1}$}  
               & \textbf{2}/1/1              & 0.1240        & 0.2141\\ 
               & \textbf{0.02}/1/1            & 0.1355        & 0.2248\\ \hline
\multirow{2}{*}{$\lambda_{2}$}  
               & 0.2/\textbf{10}/1               & 0.1390        & 0.2297\\ 
               & 0.2/\textbf{0.1}/1            & 0.1394        & 0.2294\\ \hline
\multirow{2}{*}{$\lambda_{3}$}  
               & 0.2/1/\textbf{10}              & 0.1365        & 0.2268\\ 
               & 0.2/1/\textbf{0.1}            & 0.1265        & 0.2177\\ \hline
\end{tabularx}
\caption{Loss weights for different modules.}
\label{tab:Loss weights}
\end{subtable}


\label{tab: hyperparameter experiments}
\end{table}

\subsection{Hyperparameter Experiments (RQ5)}\label{Hyper-parameter Experiments}
\begin{figure}[t]
    \centering
    \begin{subfigure}{0.228\textwidth} 
        \centering
    \includegraphics[width=\linewidth]{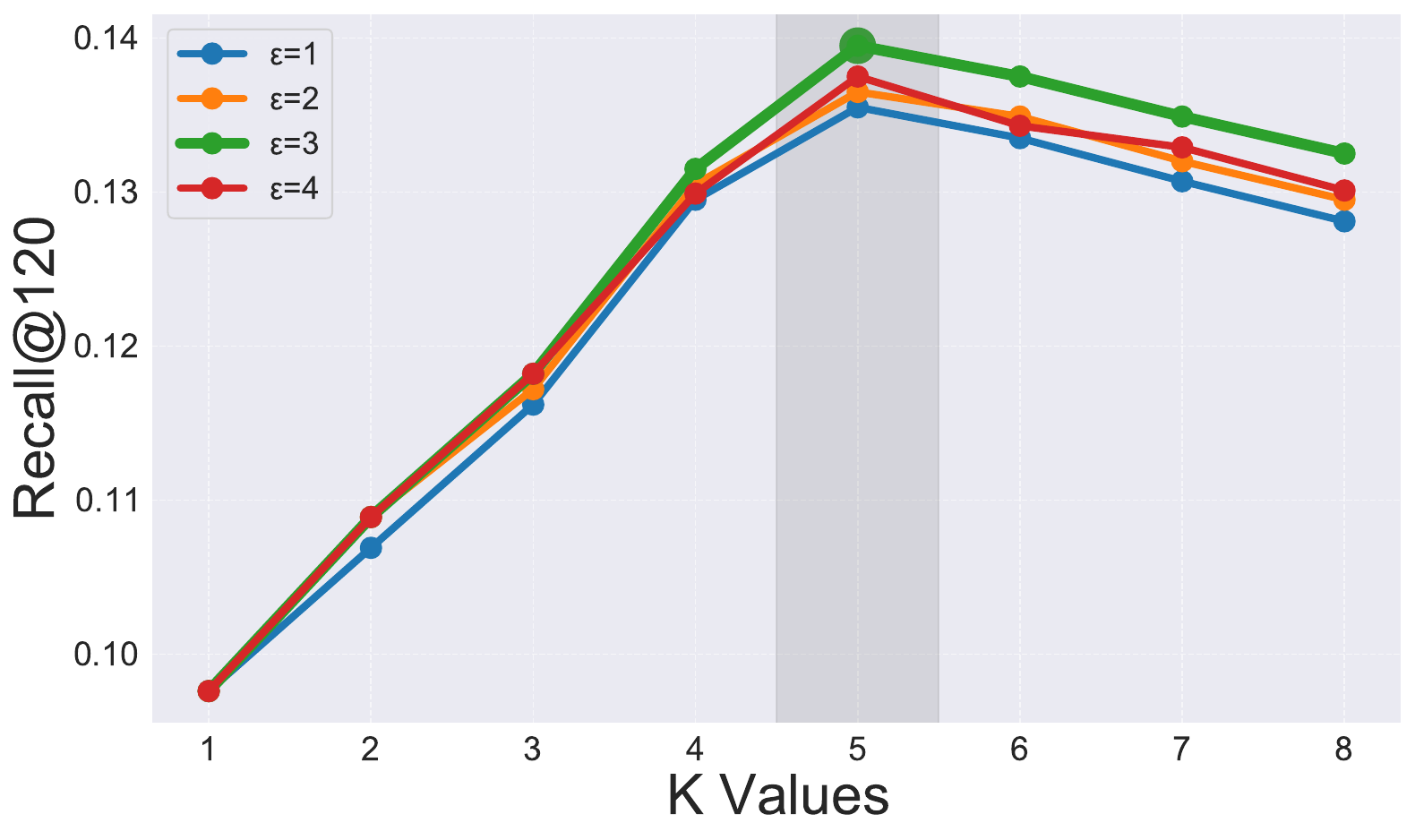} 
        \caption{\normalfont \scriptsize  Recall@120} 
        \label{fig: frequency-subcodebook}
    \label{fig: hyperparameter@120}
    \end{subfigure}
    \hspace{0.0\textwidth} 
    \begin{subfigure}{0.228\textwidth} 
        \centering
        \includegraphics[width=\linewidth]{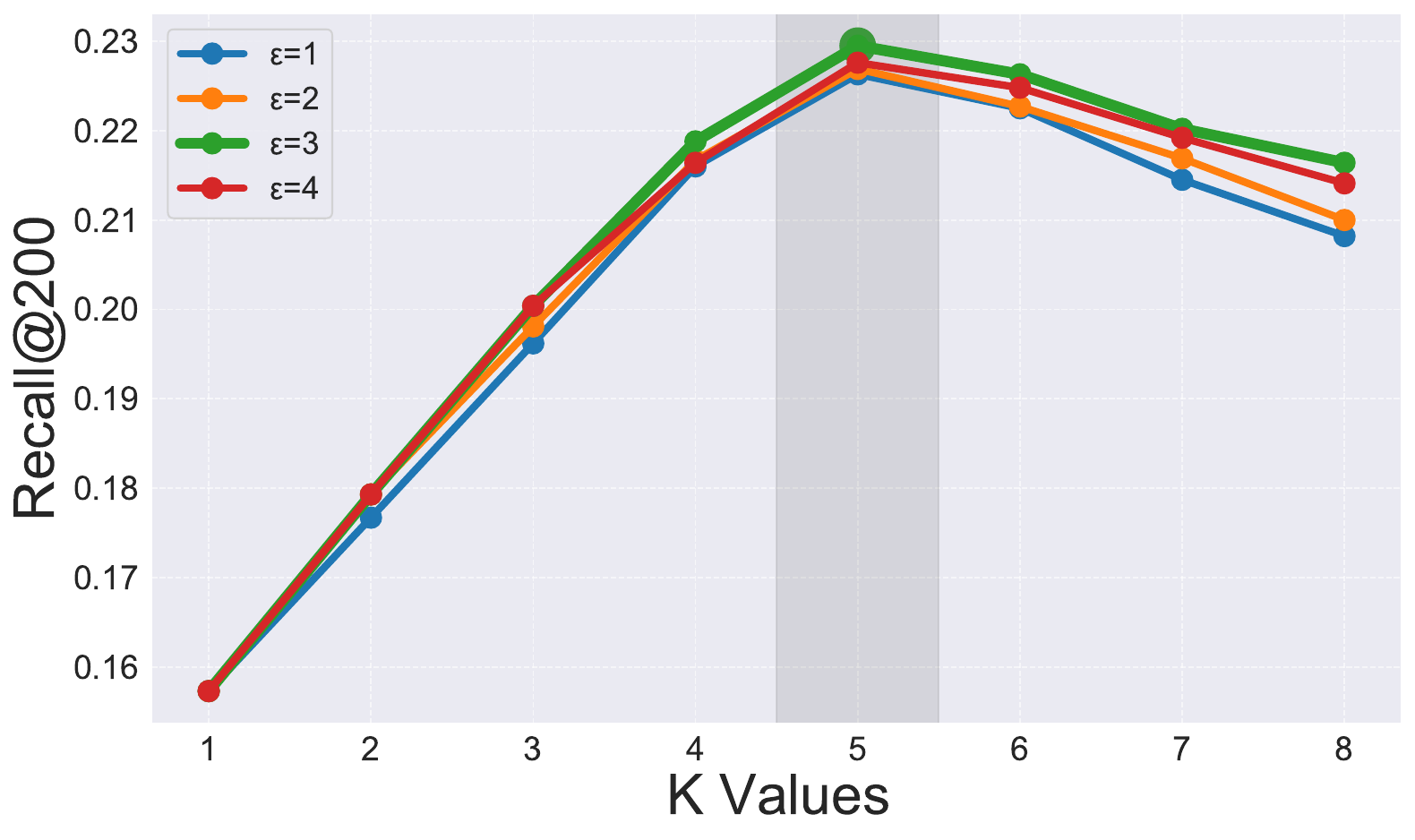} 
        \caption{\normalfont \scriptsize Recall@200} 
        \label{fig: PCA}
    \label{fig: hyperparameter@200}
    \end{subfigure}
    \caption{Hyperparameter analysis on K and $\varepsilon$.} 
    \label{fig:hyperparameter}
\end{figure}

\tightparagraph{\hspace{1.2em}\codebook{} sizes (RQ5-1/5-2).}
We evaluate different \codebook{} sizes in \cs{} 
(Table~\ref{tab:Performance of IDMM with Different Dictionary Sizes}). The trend observed in AMR@N, which reflects the interest collapse, aligns with the trend of the overall performance with respect to the Interest Dictionary size, indicating the importance of choosing an appropriate number of sub-dictionaries and a moderate overall quantization space size.

\tightparagraph{Number of GPT layers (RQ5-1/5-3).}
We vary the number of GPT layers in \ds{}. Table~\ref{tab:Effect of MIPDM Complexity} demonstrates consistent performance improvements with more layers, though the gains diminish beyond $6$. CUR@N shows a similar trend.

\tightparagraph{Task weights (RQ5-1/5-2).}
As shown in Table~\ref{tab:Loss weights}, the best performance is achieved when  $\lambda_1 = 0.2$, $\lambda_2 = 1$, and $\lambda_3 = 1$, aligning with our expectation that the update speed of the \codebook{} from \cs{} should be slower than that from \ps{}, which is directly responsible for the final recommendation. For $\lambda_1$, $\lambda_2$, and $\lambda_3$, we observe that $\lambda_2$ is relatively flexible, as \ds{} operates as a separate module without shared components, whereas maintaining a balanced ratio between $\lambda_1$ and $\lambda_3$ is crucial. 

\tightparagraph{$\varepsilon$ and $K$ (RQ5-1).}
As shown in Figure~\ref{fig:hyperparameter}, $K$ has a more significant impact. 
A small $K$ may be insufficient to capture multiple interests, while an excessively large $K$ introduces unrelated interests. 
In contrast, the effect of $\varepsilon$ follows the opposite pattern: a small value increases noise, whereas a large value reduces diversity. Therefore, well-tuned $\varepsilon$ and $K$ yield the optimal results.

\begin{table}[t] 
\small
\centering
\caption{Online A/B test on a content-sharing platform, Rednote.}
\scalebox{0.93}{\begin{tabular}{p{1.1cm}|p{1cm}|p{0.7cm}|p{0.7cm}|p{1.2cm}|p{2.cm}}
\noalign{\hrule height 1.0pt}
\textbf{Scenario} & \textbf{Duration} & \textbf{Click} & \textbf{CTR} & \textbf{Click UV} & \textbf{Next-day Active} \\
\hline
\textbf{Video} & +0.38\% & +0.37\% & +0.22\% & +0.07\% & +0.08\% \\
\hline
\textbf{Note} &	+0.26\%	& +0.51\% &	+0.32\%	& +0.08\% & +0.09\% \\

\hline
\end{tabular}}
\label{Table: onlinePerformance}
\end{table}

		

\begin{table}[t] 
\small
\centering
\caption{Computational cost in FLOPs and runtime efficiency under identical hardware settings. Latency/throughput are reported relative to the ComiRec baseline (1.00$\times$).}
\scalebox{0.88}{
\begin{tabular}{l|ccc|cc}
\noalign{\hrule height 1.0pt}
\multirow{2}{*}{\textbf{Method}} & \multicolumn{3}{c|}{\textbf{FLOPs}} & \multicolumn{2}{c}{\textbf{Runtime cost}}\\
\cline{2-6}
& \textbf{\cs} & \textbf{\ds} & \textbf{\ps} & \textbf{Latency} & \textbf{Throughput } \\
\hline
\rec{} Training & 0.36M & 3.63M & 24.25M & -- & -- \\
\rec{} Inference & -- & -- & 24.25M & 0.99$\times$ & 1.01$\times$ \\
ComiRec & -- & -- & 24.31M & 1.00$\times$ & 1.00$\times$ \\
\bottomrule
\end{tabular}}
\label{Table: additionalflopsanalysis}
\end{table}

\subsection{Online Experiments (RQ6)}
\subsubsection{Online Performance}
As shown in Table~\ref{Table: onlinePerformance}, a two-week A/B test conducted on the homepage of a content-sharing platform, Rednote (Xiaohongshu), which serves hundreds of millions of daily active users, shows statistically significant improvements across multiple recommendation scenarios and metrics at 95\% confidence level. The proposed \rec{} has been fully deployed in production since March 2025, showing its practical value.

\subsubsection{Computational Cost}
As shown in Table~\ref{Table: additionalflopsanalysis}, our method introduces a small increase in training cost compared to baselines, and such overhead is generally not a bottleneck in industrial systems. 
In deployment, where inference efficiency is more critical, \rec{} maintains comparable FLOPs, latency, and throughput to ComiRec variants that share similar inference characteristics, indicating its applicability in real-world scenarios. All methods were trained and evaluated under identical hardware settings for fairness.



\section{Conclusion}
In this paper, we propose a novel generative multi-interest recommendation framework, \rec. 
The framework introduces interest quantization and generation to address the inherent limitations of existing multi-interest methods from a new perspective. 
Theoretical and empirical analyses, together with extensive experiments and online A/B tests, demonstrate the superiority of the framework.
Furthermore, it has been deployed in production on a content-sharing platform, Rednote, confirming its practical value in industrial applications. In the future, we will explore advanced quantization and enhance the interest 
generation to better unlock the potential of the framework.

\clearpage

\bibliographystyle{ACM-Reference-Format}
\bibliography{references}
\appendix
\section*{Appendix}
\section{Math Symbols}
\label{sec:symbol}
\begin{table}[H]
\centering
\caption{Summary of math symbols.}
\label{tb:symbol}
\renewcommand{\arraystretch}{1.}
\scalebox{0.85}{
\begin{tabular}{ll}
\hline
Symbol & Meaning \\
\hline
$\mathcal{U}, \mathcal{I}$ & Set of users and items \\
$\mathcal{I}_u$ & Interaction sequence of user $u$ \\
$\hat{y_k}(\cdot)$ & Preference score of k-th user-interst representation  \\
$\bm{E}^c \in \mathbb{R}^{M_c \times d}$ & Sub-dictionary for the $c$-th interest dimension \\
$\bm{E}^*$ & Entire \codebook, $\bm{E}^* \in \mathbb{R}^{(\prod_{c}M_c) \times (Cd)}$ \\
$\bm{r}_c$ & The $c$-th residual during interest quantization \\  
$\bm{e}_{m^c}^c \in \mathbb{R}^{d} $ & The $m^c$-th interest in the $c$-th \scb\\
$\bm{e}_{m^*} \in \mathbb{R}^{Cd}$ & Quantified multi-dimensional interest embedding\\
$\bm{p}_{u},{\hat{\bm{p}_{u}}}$ & The ground-truth and predicted future interest distribution\\
$\bm{u}, \bm{v}$ & User/item embedding from the user/item tower \\
$\bm{u}_k$ & The $k$-th user-interest representation \\
\hline
\end{tabular}}
\end{table}

\section{Theoretical Proof}\label{appendix: proof}
For convenience, throughout the proofs we set $E^*\subset\mathbb{R}^{d}$ 
rather than $\subset\mathbb{R}^{Cd}$, which does not affect the generality of the results.
\begin{definition}[Voronoi Partition]\label{definition: Voronoi}
Let $E^*=\{e_1,\dots,e_{|E^*|}\}\subset\mathbb{R}^{d}$ be a finite set.
The Voronoi cell associated with $e_m$ is defined as
\[
V_m = \{x \in \mathbb{R}^d : \|x-e_m\| \le \|x-e_n\|, \ \forall n\neq m\}.
\]
The collection $\{V_m\}_{m=1}^M$ is called the Voronoi partition of $\mathbb{R}^d$.
\end{definition}

\begin{proposition}[Equivalent Characterization]\label{proposition: Equivalent}
A collection $\{V_m\}_{m=1}^M$ is the Voronoi partition induced by $E^*$ if and only if it satisfies the following properties:
\begin{enumerate}
    \item \textbf{Covering:} $\bigcup_{m=1}^{|E^*|} V_m = \mathbb{R}^d$.
    \item \textbf{Cell structure:} Each $V_m$ can be written as the intersection of finite closed halfspaces bounded by perpendicular bisectors.
    \item \textbf{Disjointness:} $V_m \cap V_n = \varnothing$ for $m\neq n$, and overlaps occur only on boundaries of measure zero.
    \item \textbf{Nearest-neighbor consistency:} Almost everywhere, $x\in V_m$ if and only if $e_m$ is the unique nearest neighbor of $x$ in $E$.
\end{enumerate}
\end{proposition}

\begin{proposition}[Proof of Interest Quantization Induces a Voronoi Partition]
Let $E^*=\{e_1,\dots,e_{|E^*|}\}\subset\mathbb{R}^{d}$  be a finite interest Dictionary. Define the nearest-neighbor quantizer
\[
q:\mathbb{R}^d \to E,\qquad 
q(x)\in \arg\min_{e\in E^*}\|x-e\|_2 .
\]
Then the quantization rule $q$ induces the Voronoi partition $\{V_m\}_{m=1}^M$ of $\mathbb{R}^d$, and almost everywhere
\[
x\in V_m \iff q(x)=e_m.
\]
\end{proposition}

\begin{proof}
Following Definition \ref{definition: Voronoi} and Proposition \ref{proposition: Equivalent}, we prove that the quantization satisfies the following properties and induces a Voronoi Partition.
\textbf{(1) Covering.}
For any $x\in\mathbb{R}^d$, the finite set $\{\|x-e\|_2: e\in E\}$ attains its minimum. Let $e_m$ be a minimizer. Then $\|x-e_m\|_2 \le \|x-e_n\|_2$ for all $n$, so $x\in V_m$. Hence $\mathbb{R}^d=\bigcup_{m=1}^M V_m$.

\medskip
\textbf{(2) Convexity of cells.}
For fixed $m\neq n$, define
\[
H_{m,n}=\{x:\ \|x-e_m\|_2 \le \|x-e_n\|_2\}.
\]
This inequality is equivalent to
\[
2\langle x, e_n-e_m\rangle \le \|e_n\|_2^2-\|e_m\|_2^2,
\]
which describes a closed halfspace bounded by the perpendicular bisector of $e_m$ and $e_n$. Thus
\[
V_m = \bigcap_{n\neq m} H_{m,n}
\]
is an intersection of finitely many closed halfspaces, hence convex and closed.

\medskip
\textbf{(3) Disjointness up to boundaries.}
If $x\in V_m\cap V_n$ with $m\neq n$, then $\|x-e_m\|_2=\|x-e_n\|_2$. The union of such equality sets
\[
B\triangleq \bigcup_{m<n}\{x:\ \|x-e_m\|_2=\|x-e_n\|_2\}
\]
is a finite union of hyperplanes (perpendicular bisectors), which has Lebesgue measure zero. Thus
\[
V_m^\circ \cap V_n^\circ = \varnothing \quad (m\neq n),
\]
and the regions are mutually disjoint except on a zero measure set.

\medskip
\textbf{(4) Consistency with nearest-neighbor quantization.}
If $x\notin B$, then there exists a unique $m$ such that $\|x-e_m\|_2<\|x-e_n\|_2$ for all $n\neq m$. By definition, $q(x)=e_m$, and simultaneously $x\in V_m$. Conversely, if $q(x)=e_m$, then the inequalities hold and hence $x\in V_m$. Therefore,
\[
x\in V_m \iff q(x)=e_m, \quad \text{for almost every } x.
\]

Combining (1)--(4), we conclude that $\{V_m\}$ is precisely the Voronoi partition induced by the sites $\{e_m\}$, and that the nearest-neighbor quantizer coincides with this partition almost everywhere.

\end{proof}

\begin{proposition}[Proof of From \codebook{} separability to retrieval separability]
If two interests are separated by at least $\Delta$, assume the fusion function $z_{\text{fusion}}(u,e)$ satisfies a local Lipschitz-type condition in its second argument, then there exists a constant $0<\alpha \le \infty$ such that their fused retrieval vectors satisfy
\[
\|z_{\text{fusion}}(u,e_m) - z_{\text{fusion}}(u,e_n)\|_2 \;\ge\; \alpha \,\Delta .
\]
In cosine-similarity maximum inner product search (MIPS), with normalized outputs, this further implies a strictly positive score-margin lower bound, which in turn upper-bounds the overlap between the Top-$N$ candidate sets of the two interests.
\end{proposition}

\begin{proof}
The argument follows from the Lipschitz-type assumption, for any user embedding $u$,
\[
\|z_{\text{fusion}}(u,e_m)-z_{\text{fusion}}(u,e_n)\|_2 \;\ge\; \alpha \|e_m-e_n\|_2 .
\]
Since $d(e_m,e_n) \ge \Delta$, we obtain
\[
\|z_{\text{fusion}}(u,e_m)-z_{\text{fusion}}(u,e_n)\|_2 \;\ge\; \alpha \,\Delta .
\]
If outputs are $\ell_2$-normalized, this separation translates into a strictly positive margin in cosine similarity. Consequently, the overlap between the corresponding Top-$N$ candidate sets is strictly upper-bounded.
The local Lipschitz-type condition can be held~\cite{li2019approximate,oberman2018lipschitz} under common architectures such as LeakyReLU with spectral norm constraints on each weight matrix.
\end{proof}


\begin{proposition}
\label{prop:reg-no-sep}
(Regularization does not imply structural separation.)  
For any finite regularization weight $\lambda$ and any continuous penalty $\mathcal{R}$, there does not exist a 
data-independent constant $\Delta > 0$ such that all global minimizers satisfy
\[
\min_{k \neq \ell}\|\mathbf{u}_k - \mathbf{u}_\ell\|_2 \;\ge\; \Delta.
\]
\end{proposition}

\begin{proof}[Proof.]
We argue by contradiction with several counterexamples to demonstrate cases where regularization fails to offer a lower-bound separation.

\begin{counterexample}[Absence under Scale invariance in dot-product retrieval]\label{counter:scale-invariance}
In dot-product or cosine retrieval, the score is invariant under rescaling:
\[
\langle \mathbf{u}, \mathbf{v} \rangle = \langle c\mathbf{u}, \tfrac{1}{c}\mathbf{v} \rangle, \quad c > 0.
\]
Thus, user embeddings can be arbitrarily shrunk while item embeddings are scaled accordingly, driving 
$\min_{k \neq \ell} \|\mathbf{u}_k - \mathbf{u}_\ell\|_2 \to 0$
, thereby circumventing distance-based regularizers.
\end{counterexample}
\begin{counterexample}[Absence under cosine normalization with Dominant Modality]\label{counter:Dominant Modality}
Assume $\ell_2$-normalized $\hat u_k,\hat v \in \mathbb S^{d-1}$. Let the positive item distribution be
\[
\hat v \sim p \cdot \mathrm{vMF}(a,\kappa_1) \;+\; (1-p)\cdot \mathrm{vMF}(b,\kappa_2),
\]
where $a,b \in \mathbb S^{d-1}$ with $\angle(a,b)>0$ and $p \gg (1-p)$.  

For the collapsed solution $\hat u_1=\cdots=\hat u_K=a$, we have
\[
\mathbb E[\hat u_k^\top \hat v \mid \hat v\sim \mathrm{vMF}(a,\kappa_1)] = \alpha(\kappa_1),
\]
which maximizes the expected score on the dominant component.  

If a fixed separation $\Delta>0$ (equivalently, angle $\theta>0$) is imposed, then some $\hat u_j$ must satisfy $\angle(\hat u_j,a)\ge \theta$, hence
\[
\mathbb E[\hat u_j^\top \hat v \mid \hat v\sim \mathrm{vMF}(a,\kappa_1)] \le \alpha(\kappa_1)\cos\theta.
\]
Thus the expected loss on the $p$-fraction dominant samples increases by at least a constant $\delta(\theta,\kappa_1)>0$. For $T$ positive samples, the cumulative gap is at least $T p \delta(\theta,\kappa_1)$.  

Meanwhile, the maximum possible gain from regularization is bounded:
\[
\Delta\mathcal R \le \lambda \binom{K}{2} R(2).
\]
For sufficiently large $T$, we obtain
\[
T p \delta(\theta,\kappa_1) > \lambda \binom{K}{2} R(2),
\]
so the collapsed solution minimizes the overall objective with
\[
\min_{k\ne \ell}\|\hat u_k-\hat u_\ell\|=0.
\]

Hence, even under cosine normalization, a finite $\lambda$ cannot guarantee a data-independent lower-bound separation.
\end{counterexample}
\begin{counterexample}[Span/Null-space indeterminacy.]\label{counter:indeterminacy}
Let all item embeddings lie in a proper linear subspace 
$S \subset \mathbb{R}^d$ of rank $r<d$. 
Decompose each user-interest vector as 
$u_k = s_k + n_k$ with $s_k \in S$ and $n_k \in S^\perp$. 

Notice that scores depend only on the projection onto $S$:
\[
\langle u_k, v_i \rangle 
= \langle s_k + n_k, v_i \rangle
= \langle s_k, v_i \rangle,
\]
since $v_i \in S$ and $n_k \perp S$. Thus, the recommendation loss $L_{\text{task}}$ 
ignores all $n_k$ components. 

Then fix $\{s_k\}$. For any $\varepsilon>0$, choose 
$\{n_k'\} \subset S^\perp$ such that 
$\max_{k\neq \ell}\|n_k' - n_\ell'\|\le \varepsilon$. 
This leaves $L_{\text{task}}$ unchanged, while making
\[
\min_{k\neq \ell}\|u_k' - u_\ell'\|
= \min_{k\neq \ell}\|(s_k+n_k')-(s_\ell+n_\ell')\|
\]
arbitrarily small, and in particular below any fixed $\Delta > 0$. 

Thereby, the null-space freedom allows embeddings to collapse in directions irrelevant to retrieval, contradicting the claim that regularization universally enforces a positive separation margin.
\end{counterexample}

Overall, Counterexample~\ref{counter:scale-invariance} demonstrates that scale invariance in the retrieval objective 
allows the pairwise distance to shrink arbitrarily without affecting overall loss. 
Counterexample~\ref{counter:Dominant Modality} further 
shows that even under normalized cosine similarity, when the positive item 
distribution is dominated by a single modality, the collapsed solution  still strictly minimizes the overall objective as the data size grows Counterexample~\ref{counter:indeterminacy} illustrates that the null-space components of user-interest representations, when item embeddings lie in a low-rank subspace, can be freely adjusted without worsening the overall objective, again reducing the minimum separation arbitrarily. 
Therefore, there does not exist a 
data-independent constant $\Delta > 0$ such that all global minimizers satisfy
\[
\min_{k \neq \ell}\|\mathbf{u}_k - \mathbf{u}_\ell\|_2 \;\ge\; \Delta.
\]
\end{proof}


\end{document}